\documentclass[conference]{IEEEtran}
\IEEEoverridecommandlockouts

\usepackage[utf8]{inputenc}
\usepackage{cite}
\usepackage{amsmath,amssymb,amsfonts}
\usepackage[linesnumbered,ruled,vlined]{algorithm2e}
\usepackage{array}
\usepackage{multirow}
\usepackage{booktabs}
\usepackage[normalem]{ulem}
\usepackage{arydshln}
\usepackage{graphicx}
\usepackage{subcaption}
\usepackage{wrapfig}
\usepackage{enumitem}
\usepackage{amsthm}
\usepackage{xcolor}
\usepackage{xspace}
\usepackage{comment}
\usepackage{textcomp}
\usepackage{url}

\def\BibTeX{{\rm B\kern-.05em{\sc i\kern-.025em b}\kern-.08em
    T\kern-.1667em\lower.7ex\hbox{E}\kern-.125emX}}

\newcommand{\oracle}{\textit{Oracle}\xspace}

\newcommand{\window}{b}

\newcommand{\argmax}{\textit{argmax}}
\newcommand{\benefit}{\textit{benefit}\xspace}
\newcommand{\query}{\textit{query}\xspace}
\newcommand{\gain}{\textit{gain}\xspace}
\newcommand{\match}{\textit{Match}\xspace}

\newcommand{\DD}{\texttt{Oracle}\xspace}
\newcommand{\subopt}{\texttt{SubOpt}\xspace}
\newcommand{\perbacco}{\texttt{pERbacco}\xspace}
\newcommand{\perbac}{\texttt{pERbac}\xspace}

\newcommand{\cora}{\texttt{Cora}\xspace}
\newcommand{\camera}{\texttt{Camera}\xspace}
\newcommand{\funding}{\texttt{Funding}\xspace}
\newcommand{\voters}{\texttt{Voters}\xspace}
\newcommand{\wdc}{\texttt{WDC-80}\xspace}
\newcommand{\synth}{\texttt{Synth10k}\xspace}

\newcommand{\blue}[1]{\textcolor{blue}{#1}}

\newtheorem{theorem}{Theorem}[section]
\newtheorem{lemma}{Lemma}[section]
\newtheorem{problem}{Problem}
\newtheorem{example}{Example}
\newtheorem{definition}{Definition}
\newtheorem{proposition}{Proposition}

\begin{document}
\title{Entity Resolution via Batched Oracle Queries}

% \author{\IEEEauthorblockN{Anonymous}}
\author{
\IEEEauthorblockN{
Lorenzo Balzotti\IEEEauthorrefmark{1},
Donatella Firmani\IEEEauthorrefmark{1},
Luca Gagliardelli\IEEEauthorrefmark{2},
Giovanni Simonini\IEEEauthorrefmark{3}
}
\IEEEauthorblockA{\IEEEauthorrefmark{1}Sapienza University of Rome, Rome, Italy\\
Emails: \{lorenzo.balzotti, donatella.firmani\}@uniroma1.it\\
ORCID: 0000-0001-6191-9801, 0000-0003-0358-3208}
\IEEEauthorblockA{\IEEEauthorrefmark{2}Università eCampus, Italy\\
Email: luca.gagliardelli@uniecampus.it; ORCID: 0000-0001-5977-1078}
\IEEEauthorblockA{\IEEEauthorrefmark{3}University of Modena and Reggio Emilia, Modena, Italy\\
Email: simonini@unimore.it; ORCID: 0000-0002-3466-509X}
}

\maketitle

\thispagestyle{plain}
\pagestyle{plain}

\begin{abstract}
We consider an oracle that processes a limited batch of records at a time and clusters those that refer to the same real-world entity.
We study how to interrogate such an oracle to resolve entities in a dataset whose size is far larger than a single batch, and where no batch is guaranteed to contain all records of any given entity. %perché? Ci possono essere entità piccole con pochi doppioni
We aim at a pay-as-you-go approach, to have full control over the costs (the number of oracle consults), while achieving the highest possible recall at every step.
% To have full control over the costs of consulting the oracle, the batches have to be presented to it progressively.
% while achieving the highest possible recall at every step.
We formally cast this problem as \textit{batched entity resolution}, prove that selecting optimal batches is NP-hard, and provide an optimal solution under a natural condition on entity sizes.
%Finally, we evaluate our approach on eight real-world datasets, demonstrating its efficacy and superiority over adapted state-of-the-art baselines.
Finally, we evaluate our approach on six datasets and show its superiority over state-of-the-art baselines.
\end{abstract}


\section{Batched Entity Resolution}

Entity Resolution (ER) is the task of identifying records that refer to the same real-world entity (in other words: that \emph{match}).
It is central to data integration and knowledge base construction, and a core challenge in large-scale data management.

State-of-the-art ER systems~\cite{DBLP:journals/csur/ChristophidesEP21, DBLP:journals/vldb/LiLSDT23, papadakis_exploration} rely on a common abstraction: a \emph{binary matching function} that takes two records as input and returns a match or non-match decision.
This pairwise abstraction underlies classical rule-based approaches, supervised learning methods~\cite{DBLP:journals/is/Mandilaras0GSTG21, DBLP:journals/cacm/DoanKCGPCMC20}, and recent neural models~\cite{DBLP:journals/vldb/LiLSDT23}.
As a result, ER pipelines are typically organized around pairwise comparisons, whose number grows quadratically with the dataset size.
For large collections, controlling this cost is therefore a primary concern, motivating extensive work on blocking~\cite{DBLP:journals/is/Mandilaras0GSTG21, javdani2019deepblock}, and prioritization techniques~\cite{DBLP:journals/tkde/WhangMG13, DBLP:conf/internetware/YuHZLS20, DBLP:journals/tkde/SimoniniPPB19, papadakis_exploration, firmani2016online} to reduce the number of comparisons.

\begin{figure}[!t]
  \centering
  \begin{subfigure}[t]{0.7\columnwidth}
    \centering
    \includegraphics[width=\linewidth]{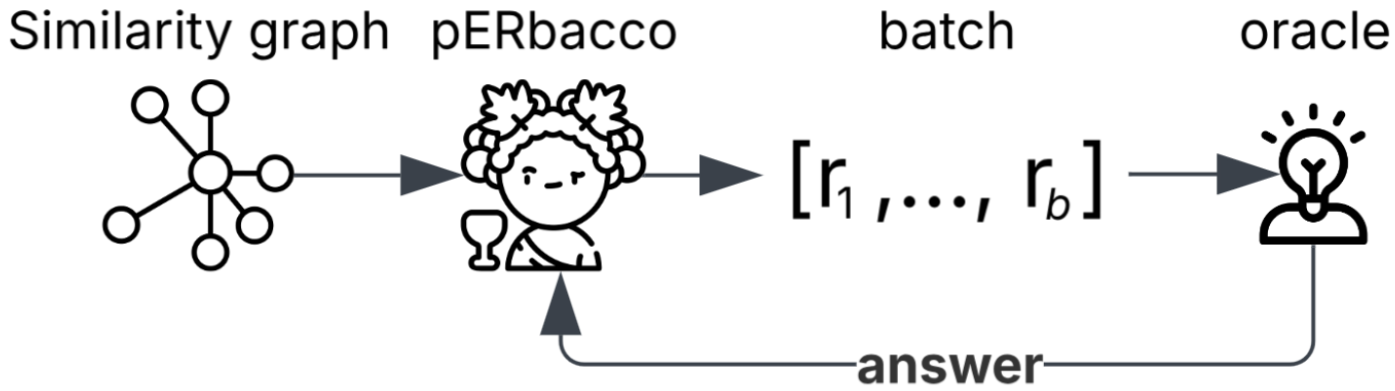}
    \caption{}
    \label{fig:intro_a}
  \end{subfigure}
  \begin{subfigure}[t]{0.58\columnwidth}
  \centering\includegraphics[width=\linewidth]{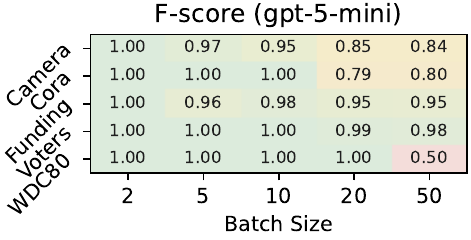}
    \caption{}
    \label{fig:intro_b}
  \end{subfigure}
  %\hfill
  \begin{subfigure}[t]{0.4\columnwidth}
    \centering
    \includegraphics[width=\linewidth]{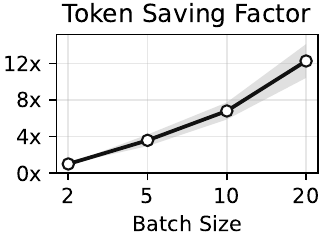}
    \caption{}
    \label{fig:intro_c}
  \end{subfigure}
  \caption{(\subref{fig:intro_a}) Entity Resolution workflow with an oracle: \perbacco selects $b$ records and passes that batch as a single invocation of the oracle.
  (\subref{fig:intro_b}) F-score of an LLM (GPT-5 mini) employed as a matching oracle, as a function of the batch size.
  (\subref{fig:intro_c}) Average token savings computed across all datasets w.r.t. the pairwise case.}
  \label{fig:intro}
\end{figure}

While pairwise matching has been the dominant abstraction for over half a century~\cite{fellegi1969theory}, many modern approaches~\cite{DBLP:conf/icml/LeeLKKCT19, DBLP:conf/coling/WangCLCHSWZ25, crowder, DBLP:conf/cvpr/SchroffKP15} operate at a higher level of granularity by jointly analyzing a \emph{batch} of items to produce a clustering.
In such models, the cost of processing a batch typically scales with the batch size $b$, while a single clustering result can resolve the match status of all $\binom{b}{2}$ record pairs induced by the batch.

For instance, \emph{Set Transformer} architectures~\cite{DBLP:conf/icml/LeeLKKCT19} can be used to learn permutation-invariant embeddings of the elements in a batch, followed by a clustering head that assigns (entity) group memberships based on joint reasoning over a set of records.
Similarly, vision-based face and identity clustering systems process bounded collections of images (e.g., photo albums) and directly group them into entity clusters without decomposing the task into independent pairwise decisions~\cite{DBLP:conf/cvpr/SchroffKP15}.
In human-in-the-loop settings, crowdsourcing-based ER approaches often present workers with small groups of records and ask them to cluster the records directly, implicitly leveraging transitivity and global context~\cite{waldo, crowder}.
Also Large Language Models (LLMs)~\cite{DBLP:journals/corr/abs-2203-02155} provide a further example of this paradigm:
an LLM can be prompted with a batch of records and asked to jointly analyze them, performing ER by returning a clustering of the input~\cite{DBLP:conf/coling/WangCLCHSWZ25}.
In Figure~\ref{fig:intro}\subref{fig:intro_b}, we randomly sampled 100 batches of close records in a similarity graph (see Section~\ref{sub:similarity_graph} for a formal definition) and computed the F-score obtained using OpenAI GPT-5 mini, with ten positive and ten negative record pairs as few-shot examples in the prompt.
We observe that the F-score is equal or very close to 1 for batch sizes up to 10, while in some datasets it degrades for larger batch sizes.

These models act as \emph{oracles} operating under inherent constraints: LLMs face context window limits and per-token costs; crowdsourcing systems must present workers with manageable tasks; neural models process fixed-size batches.
In all cases, budget constraints (monetary, computational, or human) make it infeasible to exhaustively query the entire dataset in a single step.
This naturally aligns with the \emph{pay-as-you-go} paradigm~\cite{DBLP:journals/tkde/WhangMG13}: progressively improve resolution quality while maintaining control over costs.
In Figure~\ref{fig:intro}\subref{fig:intro_c} we estimate the average token savings computed across all datasets with respect to the pairwise case.
In particular, we consider the same prompt that yields the results of Figure~\ref{fig:intro}\subref{fig:intro_b}, and a number of queries equal to the minimum required to find all matches (see Section~\ref{sub:bounds_minimum_number_of_queries}).

The main challenge in this setting becomes to determine how to allocate the available budget so as to maximize the benefit (e.g., recall) obtained at each step.
Unlike pairwise ER where the search space is well-understood~\cite{papadakis_exploration}, batched ER involves exponentially many possible batch selections, and the benefit of querying a batch depends on all previous queries.

To study this problem in a model-agnostic manner, we abstract the resolution mechanism as an oracle that, given a bounded-size set (a batch) of records, returns a partition of that set into clusters corresponding to real-world entities.
At any point during the resolution process, the task is therefore to decide which batch of records should be queried next so as to maximize the number of matches discovered for the budget spent so far.
We refer to this task as \emph{progressive batched entity resolution} and assume a consistent oracle that returns only correct match decisions. An overview of our proposed approach \perbacco is reported in Figure~\ref{fig:intro}\subref{fig:intro_a}.

We note that this abstraction deliberately idealizes the behavior of the oracle.
Handling noisy or inconsistent oracle outputs (e.g., errors arising from human annotation or imperfect models) is an important and challenging problem in its own right, and has been studied extensively in the context of crowdsourced ER~\cite{crowd_error_01,crowd_error_02,crowd_error_03,crowd_error_04}.
In this work, we focus on the algorithmic problem of allocating oracle queries under budget constraints, and leave robustness to oracle errors to future work.

%In this paper, we formalize this abstraction, analyze its computational properties, and design algorithms that exploit it.
%By moving beyond the strictly pairwise paradigm, we enable dataset-level strategies that can extract substantially more information per query under fixed budget

\smallskip
\noindent\textbf{Contributions.}
By moving beyond the strictly pairwise paradigm, this work makes the following contributions:
\begin{itemize}[leftmargin=*, itemsep=0pt, topsep=0pt]
    \item A formal definition of the \emph{progressive batched entity resolution} problem, where an oracle jointly resolves batches of records and the objective is to progressively maximize resolution quality under a limited budget of oracle calls.
    %\item A proof that selecting an optimal sequence of batches is NP-hard, establishing the intrinsic computational difficulty of allocating oracle queries at the dataset level.
    \item A proof that an optimal sequence of batches does not always exist, and that selecting the next batch by maximizing the estimated gain is NP-hard, establishing the intrinsic computational difficulty of allocating oracle queries at the dataset level. This result generalizes the known findings for pairwise ER.
    \item An approximate solution that guides batch selection to maximize the benefit (number of newly discovered matches) obtained at each oracle invocation.
    \item An experimental evaluation on real-world and synthetic datasets, showing the effectiveness of the proposed approach and its superiority over state-of-the-art baselines under comparable budgets.
\end{itemize}

\smallskip
\noindent\textbf{Paper organization.}
Section~\ref{sec:theoretical_results} presents formal problem definitions and theoretical results including NP-hardness proofs and query bounds.
Section~\ref{sec:perbacco} describes our \perbacco algorithm.
Section~\ref{sec:experiments} presents experimental evaluation.
Section~\ref{sec:related_work} discusses related work, and Section~\ref{sec:conclusions} concludes.

\section{Preliminaries and theoretical results}\label{sec:theoretical_results}

This section formalizes batched ER and its progressive version, shows that optimal progressive schedules may not exist, and relates the problem to two NP-hard selection tasks. The lower bound on the number of queries is used in Section~\ref{sec:experiments} as a stopping reference, while the sufficient condition in Subsection~\ref{sub:optimal_solution} motivates the suboptimal baseline used in the experiments.

\subsection{Problem Definition}\label{sub:problems}

The input is a dataset $R$ in which different records may represent the same entity. If two records $r,r'$ represent the same entity, we say that $r$ \emph{matches} $r'$.
Let us denote by $\sim$ the \emph{match} relation. Accordingly, for $r,r'\in R$, we write $r\sim r'$ if and only if $r$ matches $r'$. Note that $\sim$ is reflexive, symmetric, and transitive, thus it is an equivalence relation. Let us denote by $\not\sim$ the \emph{non-match} relation, which denotes two records that do not match each other. Observe that $\not\sim$ is not a relation of equivalence, as it is not reflexive.
A \emph{match edge} is a matching pair in the complete graph over $R$; this terminology is independent of the similarity graph introduced later.

We assume access to an oracle
$
\oracle : R^b \to \{\text{\emph{0}}, \text{\emph{1}}\}^{\binom{b}{2}},
$
which, given a batch $B \subseteq R$ of $b \ge 2$ records, returns a match/non-match decision for every unordered pair of records in $B$.
In other words, a single oracle query over a batch provides complete pairwise matching information within that batch.

Let $Q = (q_1, q_2,\ldots,q_T)$ be a sequence of batch queries.
We denote by $\sim_Q$ the equivalence relation induced by the match information obtained from the queries in $Q$; that is, for $r,r' \in R$, we write $r \sim_Q r'$ if the \oracle answers to the queries in $Q$ allow us to infer that $r$ matches $r'$.
Let $Repr_Q = R/\sim_Q$ be the quotient set of the record set $R$ by $\sim_Q$, and let $[r]_\sim$ denote the equivalence class of $r$.
If, based on the answers to the queries in $Q$, we can infer that two representatives $r,r' \in Repr_Q$ do not match, we write $r \not\sim_Q r'$ (and $r' \not\sim_Q r$).
Note that $\not\sim_Q$ is not an equivalence relation. Moreover, two relations $\sim_1$ and $\sim_2$ are considered equivalent if they coincide on all pairs, and we write $\sim_1 = \sim_2$.

To better illustrate this notation, let us show two opposite examples.
If $Q=\emptyset$, then $[r]_{\sim_Q} = \{r\}$ for each record $r$, and there do not exist two records $r,r'$ such that $r\not\sim_Q r'$. Conversely, if $Q$ contains all pairs in ${R\choose 2}$, then $\sim_Q = \sim$, and $r' \in [r]_{\sim_Q}$ if and only if $r'\sim r$.

We are now ready to formalize the entity resolution problem via batch queries. Given an integer $b$, we define a \emph{$b$-batch} as a batch containing at most $b$ records.

\begin{problem}[Batched Entity Resolution]\label{prob:batch_entity_resolution_BER}
    Given a set of records $R$, an oracle accessing $\sim$, an integer $b$, find a minimal sequence $Q$ of $b$-batches such that $\sim_Q = \sim$ and $\not\sim_Q =\not\sim$.
\end{problem}

Given a sequence of batches $Q = (q_1,q_2,\ldots,q_T)$, we define $\match_Q$ as the number of match edges discovered by $Q$. Moreover, for $t\leq T$, we denote by $Q_t$ the subsequence $(q_1,q_2,\ldots,q_t)$.

\begin{definition}
    Let $R$ be a set of records, $b$ an integer, $Q$ a sequence of $T$ $b$-batches. We say that $Q$ is a \emph{$b$-optimal for $R$} if
    \begin{itemize}
        \item for any sequence $Q'$ of $b$-batches, $\match_{Q_t}  \geq \match_{Q'_t}$, for all $t\leq T$,
        \item $\match_{Q}$ equals the number of match edges in $R$.
    \end{itemize}
\end{definition}

%\blue{In what follows, if we are not interested in specifying $Q$, or it is clear from the context, then we easily omit it. }

\begin{theorem}\label{th:b_optimal_solution_not_exist}
    A $b$-optimal solution for $R$ does not always exist for every integer $b$ and set of records $R$.
\end{theorem}
\begin{proof}
    Let $b=5$ and $R$ consist of the following 7 entities: $\{a_1,a_2,a_3\}$, $\{b_1,b_2,b_3\}$, $\{c_1,c_2,c_3\}$, $\{d_1,d_2\}$, $\{e_1,e_2\}$, $\{f_1,f_2\}$, $\{g_1,g_2\}$.
    %It is simply to construct two sequences $Q$ and $Q'$ of $5$-batches satisfying
    Let $Q$ be the sequence of batches consisting in $Q_1 = \{a_1,a_2,a_3,d_1,d_2\}$, $Q_2 = \{b_1,b_2,b_3,e_1,e_2\}$, $ Q_3 = \{c_1,c_2,c_3,f_1,f_2\}$, and $Q_4 = \{g_1,g_2\}$,
    and $Q'$ consisting in $Q'_1 = \{a_1,a_2,a_3,b_1,b_2\}$, $Q'_2 = \{b_1,b_3,c_1,c_2,c_3\}$, $Q'_3 = \{d_1,d_2,e_1,e_2\}$, $ Q'_4 =\{f_1,f_2,g_1,g_2\}$.
    It holds that $\match_{Q_1} = \match_{Q'_1} = 4$, $\match_{Q_2} = 8 < 9 = \match_{Q'_2}$, $\match_{Q_3} = 12 > 11 = \match_{Q'_3}$, and $\match_{Q_4} = \match_{Q'_4} = 13$. Since the number of match edges is 13, and there does not exist $Q''$ satisfying $\match_{Q''_1} \geq 4$, $\match_{Q''_2} \geq 9$, and $\match_{Q''_3} \geq 12$, then the thesis holds.
\end{proof}

Since a $b$-optimal solution cannot always be guaranteed, we aim to maximize the increase in match edges at each query in a progressive manner.

\begin{problem}[Progressive Batched Entity Resolution]\label{prob:PBER}
    Given a set of records $R$, an integer $b$, and a sequence of queries $Q$, find a $b$-batch $q$ that maximizes $\match_{Q \circ q}$, where $Q\circ q$ denotes the sequence $Q$ extended by $q$.
\end{problem}

\subsection{Batched Entity Resolution is NP-hard}\label{sub:BER_NP-hard}

To prove that Problem~\ref{prob:batch_entity_resolution_BER} is NP-hard, we first introduce the \emph{bin packing problem} (also known as the \emph{one-dimensional cutting stock problem}), which is a classical NP-hard problem~\cite{gilmore1961linear}. The bin packing problem can be stated as follows.

\begin{problem}[Bin packing problem]\label{prob:bin}
    Given a set $I$ of items, a size $s(i) \in \mathbb{N}$ for each $i \in I$, and a positive integer bin capacity $b$, find the minimum integer $\ell$ such that there exists a partition of $I$ into disjoint sets $I_1,\dots, I_\ell$ where the sum of the sizes of the items in each $I_j$ is $b$ or less.
\end{problem}

Problem~\ref{prob:batch_entity_resolution_BER} is NP-hard by reduction from the restricted bin-packing variant in which all bins must be exactly full, i.e., $k$-way number partitioning~\cite{korf2009multi}. Given such an instance, use the items as entities and the bin capacity as the batch size. Any minimal batch schedule that fully resolves these entities induces a valid exact bin packing, and vice versa.

\subsection{Minimum Number of Queries}\label{sub:bounds_minimum_number_of_queries}

In this subsection, we focus on upper and lower bounds on the minimum number of queries required to discover all match edges in Problem~\ref{prob:batch_entity_resolution_BER}.
Indeed, in a progressive approach, only the optimization of match edges is considered, not that of non-match edges.
It can be proven by the same reduction that finding this number is an NP-hard problem.
These bounds are theoretically significant, and our experiments in Section~\ref{sec:experiments} show that they provide estimates on real datasets with an error below 3\% (see Table~\ref{tab:upper_lower_bound}).
We remark that the upper bound requires the solution of an instance of the bin packing problem, while the lower bound can be computed in $O(|R|)$ time.
Therefore, the lower bound is not only efficient to compute but also a tight approximation of the true value.

%Let $\mathcal{E}$ denote the set of all entities,  $e_1, \ldots e_s$ denote the entities in non-increasing order of size. We first need to understand the minimum number of batches we need to discover all matches for an entity.

First, we determine the minimum number of $b$-batches required to discover all matches for a single entity. Consider an entity $e = \{r_1,\ldots, r_8\}$ of 8 records, and let us assume that $b = 3$. To find all matches in $e$, four $3$-batches are required.
Indeed, we can start by querying the \oracle with the $3$-batches $\{r_1,r_2,r_3\}$ and $\{r_4,r_5,r_6\}$. After this, we have only two records not already visited, $r_7$ and $r_8$,
but at least one record from the first $3$-batch must be compared with one from the second $3$-batch. Thus, the following batch can be $\{r_1,r_4,r_7\}$. After this, the standard ER equivalence closure tells us that all records $r_1,\ldots,r_7$ match each other. Finally, we need a last $3$-batch containing a record among $r_1,\ldots,r_7$ and $r_8$. No solution uses fewer than three $3$-batches full and one with only two records.

We want to generalize the example. We need some definitions.

\begin{definition}\label{def:recursive_modulo}
    Let $x$ and $b$ be two positive integers. We define the \emph{recursive rest of $x$ modulo $b$} as
    \begin{equation}
        r^{rec}_b(x) =
        \begin{cases}
            x, &\text{if $x < b$,}\\
            r^{rec}_b(n+x'), &\text{if $x = nb+x'$, for $x'<b, n\geq1$.}
        \end{cases}
    \end{equation}
We define the \emph{recursive quotient of $x$ modulo $b$} as
    \begin{equation}
        q^{rec}_b(x) =
        \begin{cases}
            0, &\text{if $x<b$,}\\
            n+q^{rec}_b(n+x'), &\text{if $x = nb+x'$, for $x'<b, n\geq1$.}
        \end{cases}
    \end{equation}
\end{definition}

%Note that $b>0$ and $x>0$ imply $r^{rec}_b(x)\geq1$.
We observe that $r^{rec}_b(x) = 0$ if and only if $x=0$, for any integer $b$. Let us apply Definition~\ref{def:recursive_modulo} to the previous example.
We have $b=3$ and $x = 8$, where $x$ denotes the cardinality of the entity $e$. We have $r^{rec}_b(x) = r^{rec}_3(8) = 2$ and $q^{rec}_b(x) = q^{rec}_3(8) = 3$.
To find all matches in $e$ we needed exactly three batches completely full of records in $e$ and one batch with exactly two records in $e$. So, before finding all matches in $e$, we first reduced it to a set of only 2 records. We want to formalize this concept.

Let $\mathcal{E}$ denote the set of entities in $R$, i.e., the set of distinct real-world entities in $R$. By our notation $\mathcal{E}=R/\sim$. Additionally, for an entity $e\in\mathcal{E}$ let $|e|$ indicate its cardinality, i.e., $|e| = |\{r\in R \,|\, r\sim e\}|$. We can now formally define the \emph{representatives} of an entity.

\begin{definition}\label{def:reduction}
    Let $R$ be a set of records, $e$ an entity in $\mathcal{E}$, and $Q$ a sequence of queries. We say that \emph{$e$ is reduced to $\ell$ representatives by $Q$} if there exist $e_1, e_2,\ldots, e_\ell$ records in $R$ such that, for every $r\sim e$, there is a unique $1\leq i\leq \ell$ satisfying $r \sim_Q e_i$.
\end{definition}

The records $e_1, e_2,\ldots, e_\ell$ are called \emph{representatives of $e$}. We generalize the previous example in the following lemma.

\begin{lemma}\label{lemma:representatives}
    Let $R$ be a set of records, $e$ an entity in $\mathcal{E}$, and $b$ an integer.
    To reduce $e$ to $r^{rec}_b(|e|)$ representatives, at least $q^{rec}_b(|e|)$ $b$-batches are required.
    Moreover, if exactly $q^{rec}_b(|e|)$ $b$-batches are used, then each of these batches consists exclusively of representatives of $e$.
\end{lemma}
\begin{proof}
We proceed by induction on $|e|$, the size of $e$. If $|e|< b$, then $q^{rec}_b(|e|) = 0$ and $r^{rec}_b(|e|) = |e|$, so the statement holds by taking each record in $e$ as its own representative.
Assume the statement is true for all  $|e|\leq N-1$, and let us prove it for $|e| = N$.

Let $[e]_\sim=\{r_1, r_2,\ldots, r_N\}$, and let $e'$ be a dummy entity satisfying $[e']_\sim = \{r_1, r_2,\ldots, r_{N-1}\}$. By induction, $e'$ can be reduced to $r^{rec}_b(|e'|)$ representatives using $q^{rec}_b(|e'|)$ $b$-batches. Now there are two cases: either $r^{rec}_b(|e'|) < b-1$, or $r^{rec}_b(|e'|) = b-1$. In the first case, it holds that $r^{rec}_b(|e|) = r^{rec}_b(|e'|) +1$. Thus, no additional $b$-batches are required: it suffices to keep the same representatives of $e'$ and consider $r_N$ as its own representative. Being $q^{rec}_b(|e|) = q^{rec}_b(|e'|)$, the claim is true in this case.

In the second case, we have $r^{rec}_b(|e|) = 1$, and a $b$-batch containing $r_N$ together with all representatives of $e'$ is required. After this batch, $e$ is reduced to a single representative. Since $q^{rec}_b(|e|) = q^{rec}_b(|e'|) + 1$, the claim holds.
\end{proof}

\begin{comment}
    \blue{TUTTA QUESTA PARTE BLU PUO' ESSERE ELIMINATA
Let us introduce a crucial example, showing that sometimes it is convenient creating batches with records from different entities. We point out that a similar variant of the bin packing problem, called \emph{bin packing with item fragmentation}, is presented in \cite{menakerman2001bin}, where a bin can be split and each split bin is augmented by one. To the best of our knowledge, our problem has not been studied.}

\begin{example}\label{ex:5_6_7}
\blue{Let $e_x, e_y, e_z$ three entities composed by 5,6 and 7 records, respectively. To find all matches in these entities, we need to reduce each one to one representative. Assume that $b = 10$. Then we cannot put all the representatives of any two entities into the same $b$-batch. We can use Problem~\ref{prob:bin} for finding a solution, and, in this case, we obtain 3 $b$-batches (bins) as solutions. However, we can obtain the reduction to one representative for each entity by using only 2 $b$-batches. For instance, we create the first $b$-batch with all the representatives of $e_x$ and 5 representatives of $e_y$. After this, $e_x$ is reduced to one representative and $e_y$ to 2. Thus we can insert these last two representatives of $e_y$ with all the 7 representatives of $e_z$ in the second $b$-batch.
    Another example. If $e_x, e_y, e_z$ had been reduced to $6,7,7$ representatives, respectively, then the minimum number of $b$-batches would have been solved by Problem~\ref{prob:bin}.}
\end{example}
\end{comment}

\begin{theorem}\label{th:query_min}
    Given a set of records $R$ and an integer $b$, let $\Phi_{b}$ be the minimum number of $b$-batches required to discover all matches in $R$. Then
    \begin{equation}\label{eq:phi_b}
    \Omega + \frac{1}{b}\sum_{e\in\mathcal{E}_1} r^{rec}_b(|e|)\leq \Phi_b\leq \Omega + BPP(\mathcal{E}_1)
    \end{equation}

    where $\Omega = \sum_{e\in\mathcal{E}} q^{rec}_b(|e|)$, $\mathcal{E}_1 =\{e\in\mathcal{E} \, | \, r^{rec}_b(|e|)>1\}$, and $ BPP(\mathcal{E}_1)$ is the solution of Problem~\ref{prob:bin} for $\mathcal{E}_1$, in which the capacity is $b$, the items are the entities in $\mathcal{E}_1$, and the size of each entity $e$ is $r^{rec}_b(|e|)$.
\end{theorem}
\begin{proof}
By Lemma~\ref{lemma:representatives}, at least $\Omega$ $b$-batches are required to reduce every entity $e$ to $r^{rec}_b(|e|)$ representatives. The remaining representatives of the entities in $\mathcal{E}_1$ must still be joined, which requires at least $\frac{1}{b}\sum_{e\in\mathcal{E}_1}r^{rec}_b(|e|)$ additional batches. Conversely, using $BPP(\mathcal{E}_1)$ batches, the remaining representatives can be packed so that each entity is reduced to a single representative, proving the upper bound.
\end{proof}

%We observe that the lower bound in Equation~\ref{eq:phi_b} can be computed in $O(|R|)$ time, while the upper bound requires solving an instance of Problem~\ref{prob:bin}, which is NP-hard. Nevertheless, our experiments on real datasets (see Subsection~\ref{sub:???}) show that the lower bound differs from the upper bound by less than 3\%.

\subsection{Hardness of Batch Selection}\label{sub:PBER_NP-hard}

From now on, let $\mathcal{E} = (e_1, e_2,\ldots)$ denote the sequence of entities ordered by cardinality, i.e., $|e_1|\geq |e_2| \geq\ldots$.
Given a sequence of batches $Q = (q_1,q_2,\ldots,q_T)$ and a batch $q$, we recall that $Q\circ q$ denotes the sequence $(q_1,q_2,\ldots,q_T, q)$, i.e., $Q$ extended by $q$.
Moreover, let $\Delta_q^Q$ denote the number of match edges discovered by querying the \oracle with $q$, assuming the queries in $Q$ have already been asked, formally, $\Delta_q^Q = \match_{Q\circ q} - \match_Q$.
For a record $r$, let $|r|_Q$ denote the cardinality of $[r]_{\sim_Q}$, i.e., the number of records known to match $r$ after the queries in $Q$. It holds
\begin{equation}\label{eq:Delta_q}
    \Delta^Q_q  = \sum_{(r,r')\in \binom{q}{2}} \mathbf{1}_{\{r\sim r'\}}|r|_Q\cdot |r'|_Q.
\end{equation}

%Note that, if all $\binom{b}{2}$ pairs of records in $q$ are matching pairs, then $\Delta_q^Q = \sum_{(r,r')\in \binom{q}{2}} |r|_Q\cdot |r'|_Q$. %This observation leads to the following definition that is used in Lemma~\ref{lemma:gain} to determine the gain of match edges, given priority to larger entities.

\begin{definition}\label{def:gain}
    Let $R$ be a set of records, $Q$ a sequence of batches, and $b$ an integer. For two distinct elements $r,r'$ in $R/\sim_Q$ we define
\begin{equation}\label{eq:gain}
    \gain^b_Q(r,r') =
    \begin{cases}
        0, & \text{if } r\not\sim r',\\
        %0, & \text{if } r\sim_Q r',\\
        |r|_Q\cdot |r'|_Q \cdot (1+\frac{1}{i\cdot \binom{b}{2}\cdot|e_1|^2}), &\text{if $r,r'\in e_i$.}
    \end{cases}
\end{equation}
Moreover, let $\textit{Gain}^b_Q$ denote the complete graph on $R/\sim_Q$ in which the edges' weights are given by $\gain^b_Q$.
\end{definition}

The importance of the coefficient $(1+\frac{1}{i\cdot \binom{b}{2}\cdot|e_1|^2})$ is necessary to prioritize larger entities without affecting the number of new match edges discovered with a $b$-batch, as formalized in the following lemma.

\begin{lemma}\label{lemma:gain}
Let $R$ be a set of records, $Q$ a sequence of batches, and $b$ an integer. Let $q$ be a $b$-batch in $R/\sim_Q$, then
\begin{equation}\Delta^Q_q  \leq \sum_{(r,r')\in \binom{q}{2}} \gain^b_Q(r,r') < \Delta^Q_q  + 1
\end{equation}
\end{lemma}
\begin{proof}
    By Equation~\ref{eq:Delta_q} it suffices to prove that $\Delta^Q_q  < \binom{b}{2}\cdot|e_1|^2$. We have exactly $\binom{b}{2}$ pairs in $q$, and each pair can discover at most $\frac{|e_1|(|e_1|-1)}{2} < |e_1|^2$ match edges. The thesis follows.
\end{proof}

To proceed, we first introduce a classical graph problem, whose NP-hardness was established in \cite{Feige2001}.

\begin{problem}[Heaviest Subgraph Problem (HSP)]\label{prob:k-heaviest}
    Given an edge-weighted graph $G$ and an integer $k$, find a subset $S$ of $k$ vertices of $G$ such that the sum of the edges' weight in $G[S]$ is maximized.
\end{problem}

Let us define $HSP(G,k)$ as a solution for Problem~\ref{prob:k-heaviest} for graph $G$ and integer $k$.  In the following proposition, we describe the solution of Problem~\ref{prob:PBER}.

\begin{proposition}\label{prop:max_gain_with_HSP}
    Let $R$ be a set of records, $Q$ a sequence of queries, and $b$ an integer. Let $q$ be a $b$-batch in $R/\sim_Q$ that is a solution of HSP for $\textit{Gain}^b_Q$. Then $\Delta^Q_q \geq \Delta^Q_{q'}$ for every $q'$ $b$-batch in $R/\sim_Q$.
\end{proposition}
\begin{proof}
    Let $H$ be the graph in $R/\sim_Q$ in which the weight of the edge between records $r, r'$ is $|r|_Q\cdot|r'|_Q$, for every $r,r'\in  R/\sim_Q$ satisfying $r\sim r'$, and 0 otherwise. By Equation~\ref{eq:Delta_q}, the maximum $\Delta^Q_{q}$ is obtained by a solution $q$ of HSP on $H$. By Lemma~\ref{lemma:gain}, $q$ is a solution of HSP for $\textit{Gain}^b_Q$ if and only if $q$ is a solution of HSP for $H$. The thesis follows.
\end{proof}

In practice $\sim$ is unknown, so any algorithm must select batches by maximizing an \emph{estimate} of the gain (such as the benefit introduced in Section~\ref{sec:perbacco}), which can form an arbitrary edge-weighted graph. The following theorem shows that this selection problem is NP-hard.

\begin{theorem}\label{th:progressive_is_NP_hard}
    Let $R$ be a set of records, $Q$ a sequence of queries, and $b$ an integer. Given a weight function $\hat{g}$ on the pairs of $R/\sim_Q$ estimating $\gain^b_Q$, finding a $b$-batch $q$ that maximizes $\sum_{(r,r')\in\binom{q}{2}}\hat{g}(r,r')$ is NP-hard.
\end{theorem}
\begin{proof}
    We reduce from Problem~\ref{prob:k-heaviest}. Given an instance $(G,k)$ of HSP with non-negative weights, we construct the instance in which $R = V(G)$, $Q=\emptyset$, $b=k$, and $\hat{g} = w$. Since $Q = \emptyset$, every record is its own representative; thus the $b$-batches of $R/\sim_Q$ coincide with the $k$-subsets of $V(G)$, and a $b$-batch maximizing the total estimated gain is exactly a solution of $HSP(G,k)$. The thesis follows from the NP-hardness of Problem~\ref{prob:k-heaviest}~\cite{Feige2001}.
\end{proof}

%\subsection{Optimal Solution in Some Special Cases}\label{sub:optimal_solution}
\subsection{Sufficient Conditions for the Existence of Optimal Solution}\label{sub:optimal_solution}

%In this subsection, we prove that in some special cases, the optimal solution exists and we can describe it. Ahe approach described is a generalization of the one described in~\cite{firmani2016online} for the pair case.

In this subsection, we show that, under certain conditions, an optimal solution exists and describe it. The proposed approach generalizes the method introduced in~\cite{firmani2016online} for the pairwise case.

\begin{theorem}\label{th:optimal}
    Let $R$ be a set of records, and $b$ an integer. Let $Q=(q_1,q_2,\ldots)$ be a succession of $b$-batches defined as
    \begin{equation}
        q_t =
        \begin{cases}
            HSP(\textit{Gain}^b_{\emptyset},b) & \text{if $t=1$,}\\
            HSP(\textit{Gain}^b_{Q_{t-1}},b) & \text{if $t> 1$}.
        \end{cases}
    \end{equation}

If $r^{rec}_b(|e|) = 1$ for each entity $e$ in $R$, then $Q$ is $b$-optimal for $R$.
\end{theorem}
\begin{proof}
By Lemma~\ref{lemma:representatives}, each entity $e$ with $r^{rec}_b(|e|)=1$ can be reduced to one representative using $q^{rec}_b(|e|)$ batches containing only representatives of $e$. Thus, as in the pairwise case~\cite{firmani2016online}, the optimal order is to reduce entities in non-increasing size because match edges grow quadratically with entity size. The coefficient in Equation~\ref{eq:gain} enforces this order without changing the number of newly discovered match edges, and Proposition~\ref{prop:max_gain_with_HSP} ensures that each selected batch maximizes the immediate gain among schedules with this order.
\end{proof}

We observe that when $b=2$, $r^{rec}_b(|e|) = 1$ for every entity $e$, so an optimal solution for pairwise batches always exists by Theorem~\ref{th:optimal}. In addition, the optimal solution described in Theorem~\ref{th:optimal} coincides with the one described in~\cite{firmani2016online}.

\begin{figure*}[t]
  \centering
  \includegraphics[width=0.95\textwidth]{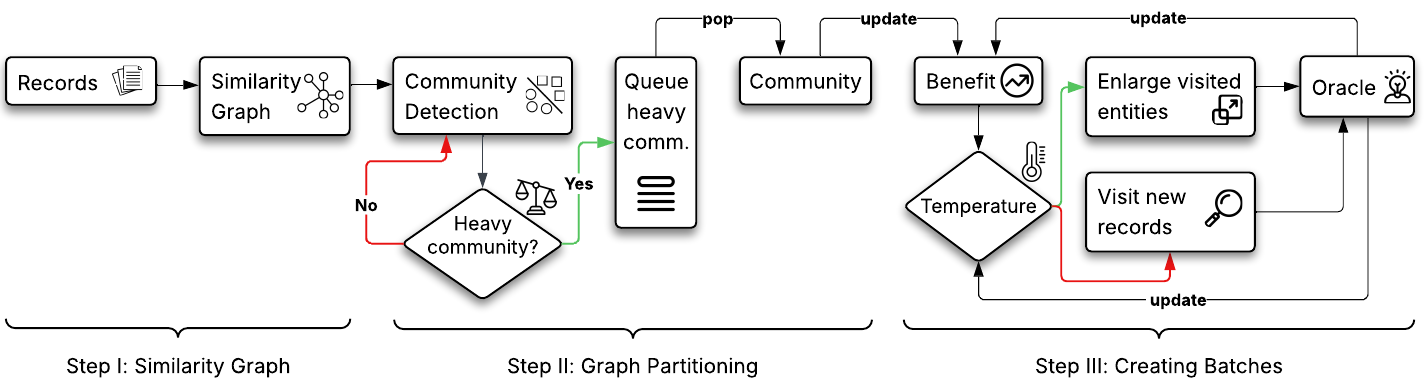}
  \caption{Schematic representation of the \texttt{pERbacco} algorithm, highlighting its core components.}
  \label{fig:schema_perbacco}
\end{figure*}

\section{Proposed approach: pERbacco}\label{sec:perbacco}

In this section, we present our approach for \underline{p}rogressive \underline{E}ntity \underline{R}esolution via \underline{ba}t\underline{c}hed queries and \underline{co}mmunity detection (called \emph{pERbacco}).
It consists of three steps: computing a similarity graph $G$ on $R$, partitioning $G$ by recursively calling a Community Detection Algorithm (CDA), and adaptively deciding whether to enlarge already visited entities or visit new records belonging to the same community. The first two steps are static, while the third depends on the answers obtained from the \oracle. A schematic overview of \perbacco is shown in Figure~\ref{fig:schema_perbacco}.

Let us briefly analyze these steps one by one. There are several methods in the literature to compute a similarity graph $G$; we refer to Subsection~\ref{sub:similarity_graph} for further details. We emphasize that our approach does not rely on any specific method.

In the second step, we partition the similarity graph $G$ into \emph{heavy communities}. All records that do not belong to any heavy community form the \emph{residual graph}. A heavy community is a subgraph of $G$ whose density is greater than a fixed threshold (see Subsection~\ref{sub:graph_partitioning}). Since the weights of match edges in $G$ are expected to be higher than those of non-match edges, a heavy community should mainly consist of match edges. We define the \emph{intra-recall} as the ratio between the number of match edges that belong to at least one heavy community and the number of match edges in $G$. The desiderata for heavy communities are:
\begin{itemize}
\item a high number of match edges within each heavy community,
\item a low number of match edges between different heavy communities,
\item an intra-recall close to the overall recall of $G$,
\item community sizes roughly comparable to $\window$.
\end{itemize}

By following the structure induced by these heavy communities, the discovery of match edges should be accelerated. The algorithm for this second step is explained in Subsection~\ref{sub:graph_partitioning}, and experiments are in Subsection~\ref{sub:communities_experiments}.

In the last step, we decide whether it is more convenient to try to enlarge already visited entities or visit new records belonging to the same community. The main idea is that visiting new records should discover fewer match edges (see Table~\ref{tab:batch_comparison} and Subsection~\ref{subsub:current_community_batches}). So, if by trying to enlarge already visited entities we discover a few edges, then in the next query we visit new records of the current community. The decisions are driven by an evolving parameter called \emph{temperature}, see Subsection~\ref{subsub:temperature}. This temperature is linked to the concept of  \emph{benefit}, where the benefit, introduced in~\cite{firmani2016online} and formally defined in Subsection~\ref{subsub:benefit}, is a measure of the expected gain in match edges -- and thus, gain in recall.  Note that this last step is repeated for each query.

\subsection{Similarity Graph}\label{sub:similarity_graph}

A similarity graph is a weighted graph where the vertices are records in $R$, and an edge weight is an uncalibrated \emph{similarity score}: higher weights indicate stronger evidence that two records match.

%There is a wide literature about methods for computing similiraty graphs... \blue{continuare il bla bla bla}

The construction of similarity graphs has been extensively studied in the context of entity resolution and data integration.
Classical approaches compute pairwise similarities by combining attribute-level similarity functions, such as string edit distances, token-based measures (e.g., Jaccard or TF–IDF), and numeric distance functions, often using manually defined rules or weighted combinations~\cite{DBLP:books/daglib/0030287, DBLP:journals/pvldb/PaulsenGD23}.
More advanced methods adopt learning-based models, where similarity functions or matching scores are learned from data \cite{li2020deep, peeters2023using, wang2023sudowoodo}.
%While these approaches can improve matching accuracy, they typically introduce additional computational complexity and, in some cases, require training data.

Independent of the specific similarity model, ER systems must address the quadratic complexity of pairwise comparisons.
For this reason, similarity computation is almost always coupled with blocking and indexing techniques, whose goal is to restrict similarity evaluation to promising candidate pairs.
A large body of work investigates blocking methods, which are commonly integrated into both traditional and learning-based ER pipelines \cite{DBLP:journals/tkde/Christen12, papadakis2016scaling, DBLP:journals/is/GagliardelliPSBP24, thirumuruganathan2021deep, javdani2019deepblock}.

In this work, we assume that a similarity graph $G$ on $R$ is given, and let $w : E(G) \to [0,1]$ be the edge-weight function; without loss of generality, we assume that the maximum weight is 1.
Our approach is independent of the specific technique used to compute similarities, and can operate with any method that assigns edge weights reflecting matching evidence.

%From now on, we assume that we have a similarity graph $G$ on $R$, and let $w : E(G) \to [0,1]$ the weight function on edges; without loss of generality, we assume that the maximum weight is 1.

\paragraph{The role of the similarity graph} The similarity graph serves as the starting point for both Problem~\ref{prob:batch_entity_resolution_BER} and Problem~\ref{prob:PBER}. Consequently, different similarity graphs yield different performance for our method and competing approaches. In our experiments, we use a uniform procedure to construct the similarity graph across all datasets, see~\cite{DBLP:journals/is/Mandilaras0GSTG21}. The resulting graphs show standard characteristics, with recall above 0.9 and precision on the order of 0.03 or lower. Regarding the time required to compute the similarity graph, it ranges from a few seconds for small datasets to a few minutes for the largest ones. This computation represents only a small fraction of the overall procedure.

Our method is robust to variations in the precision and recall of the similarity graph; as expected, higher-quality similarity graphs lead to improved performance.

\subsection{Graph Partitioning}\label{sub:graph_partitioning}

We need to balance the weight and the size of the communities. Indeed, splitting a community into two or more subcommunities may increase the average weight, but we might lose some edges that match. We find communities via a Community Detection Algorithm (\emph{CDA}), and we say that a community is \emph{indivisible for CDA} if the CDA cannot divide it into more subcommunities. Finally we define $\textit{min\_size} = \max(\window,10)$, and we do not consider \emph{small} communities, i.e., communities with less than $\textit{min\_size}$ records.

\begin{figure}[h]
  \centering
  \vspace{-5pt}
  \includegraphics[width=0.75\columnwidth]{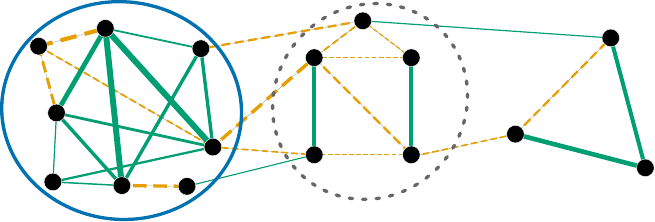}
  \caption{A similarity graph with match (solid green) and non-match (dotted orange) edges. Heavy and light communities are marked with solid blue and dotted gray, respectively.}
  \label{fig:community}
  \vspace{-10pt}
\end{figure}

In Figure~\ref{fig:community}, there is a  similarity graph $G$, where edge thickness denotes weights. A heavy community is highlighted in blue, and a light (i.e., not heavy) one in gray. Recursively, light communities are sent to the CDA to further split them until they become either heavy or indivisible. The records that do not belong to any heavy community form the residual graph. Note that in Figure~\ref{fig:community} the only one heavy community contains a relevant part of the match edges.

After the computation of the heavy communities, we order them by their weights (see Equation~\ref{eq:weight}). We denote this order by $\mathcal{C} = (c_1, c_2, \ldots, c_r)$. We will analyze the communities by following this order; the pseudocode is shown in Algorithm~\ref{alg:community}. In the whole paper and in all presented algorithms, ties are broken randomly.  An empirical analysis on the CDAs, the heaviness threshold $\lambda_w$, and how to obtain the desiderata described at the beginning of this section are in Subsection~\ref{sub:communities_experiments}.

%{\color{blue}We expect to have more match edges in the first communities, and few match edges intra-community. The last two concepts can be resumed as: if $i >> j$, then we expect that $\recall(c_i) > \recall(c_j)$, and if $i\neq j$, then $\recall(c_i\cup c_j) \approx  \recall(c_i) + \recall(c_j)$.}

We define the \emph{weight} of a community $c$ as
\begin{equation}\label{eq:weight}
     w(c) = \sum_{e\in c} w(e),
\end{equation}
The main idea is that the higher $w(c)$, the higher the number of expected match edges. Moreover we define the \emph{density} $\rho(c)$ of a community $c$ as the classical edge-weighted density:

\begin{equation}\label{eq:omega}
     \rho(c) = \frac{w(c)}{|c|(|c|-1)/2}.
\end{equation}

%l'ultima community NON può essere i restanti, sennò non funziona \perbacco

\begin{algorithm}[h]
\caption{Communities $(G, \lambda_w, CDA)$}\label{alg:community}
\KwIn{a graph $G$, a threshold $\lambda_w$, and a community detection algorithm $CDA$}
\KwOut{a queue of heavy communities of $G$}
{$\textit{light\_communities} = [V(G)]$\;
$\textit{heavy\_communities} = []$\;
%$atoms = []$\;
\While{$\textit{light\_communities}\neq []$}{
    $c = pop(\textit{light\_communities})$\;
    $\textit{sub\_communities} = CDA(c)$\;
    \For{each $c'\in \textit{sub\_communities}$}{
            \If{$\rho(c') \geq \lambda_{w}$}{
                add $c'$ to $\textit{heavy\_communities}$\;
            }
            \If{$\rho(c') < \lambda_{w}$ and $c'$ is not indivisible for CDA}{
                add $c'$ to $\textit{light\_communities}$\;
            }
        }
}
Discard from $\textit{heavy\_communities}$ all the communities having less than $\textit{min\_size}$ records\;
Sort the communities in $\textit{heavy\_communities}$ w.r.t. $w$\;
\Return{$\textit{heavy\_communities}$}
}
\end{algorithm}

\subsection{Creating Batches}

This subsection details the adaptive batch-selection step. We score candidate representative pairs by \emph{benefit}, use a greedy HSP routine to form high-benefit batches, and alternate between \emph{current-batches} that can enlarge known entities and \emph{community-batches} that introduce unqueried records. The temperature parameter controls this alternation.

\paragraph{From theory to algorithm.}
Proposition~\ref{prop:max_gain_with_HSP} shows that the ideal next batch is an HSP solution on the gain graph, but that graph depends on the unknown ground truth and HSP is NP-hard. The algorithm below is the observable counterpart of that characterization: it replaces the hidden gain graph with the benefit graph induced by similarity scores and previous oracle answers, uses \textit{GreedyHS} as the HSP surrogate, and restricts exploration through heavy communities to focus the search on dense regions of the similarity graph.

\subsubsection{Benefit}\label{subsub:benefit}

As we query the \oracle, the number of records that are representative of more than one record increases. So it may be helpful for the progressive recall to create batches composed by the representatives that are likely to match. We need a measure that expresses this.

Given a sequence of queries $Q$ and a pair of records $r$ and $r'$, we define the $\emph{benefit}_Q(r,r')$ as the potential gain in matches obtained by asking the \oracle whether $r$ and $r'$ match, as done in~\cite{firmani2016online}. We want $\benefit_Q(r,r')$ to be proportional to $|r|_Q\cdot|r'|_Q$, which corresponds to the number of match edges that would be discovered if $r$ matches $r'$. Moreover, $\benefit_Q(r,r')$ should increase with the evidence that $r$ and $r'$ match. We capture this evidence through an uncalibrated \emph{matching coefficient} $\mu_Q(r,r')$.

The formula for $\benefit_Q(r,r')$ follows:

\begin{equation}\label{eq:benefit}
    \benefit_Q(r,r') =
    \begin{cases}
        0, & \text{if } r\not\sim_Q r',\\
        \mu_Q(r,r') \cdot |r|_Q\cdot |r'|_Q, &\text{otherwise.}
    \end{cases}
\end{equation}

If $Q$ is clear from the context, then we omit it. Note that the term $|r|_Q\cdot|r'|_Q$ in Equation~\ref{eq:benefit} corresponds to the number of match edges discovered if $r$ matches $r'$. %So the benefit balances the crossing benefit with the cardinality of the two entities represented by $r$ and $r'$.

\paragraph{Matching coefficient}\label{subsub:matching_coefficient}

%As just stated, in order to define the benefit between two records $r$ and $r'$, we introduced the \emph{matching coefficient} $\mu_Q(r,r')$ between $r$ and $r'$.
Given two records $r$ and $r'$, there is no unique way to define $\mu_Q(r,r')$, but higher values should indicate stronger matching evidence. We therefore use $\mu_Q(r,r')$ only as an uncalibrated likelihood proxy.

We define $\textit{Cross}_Q(r,r') =\{(u,v)\in E(G) \ | \ u\in [r]_{\sim_Q}$ and $v\in [r']_{\sim_Q}\}$. Basically, $\textit{Cross}_Q(r,r')$ is composed by all edges in $G$ having one endpoint in $[r]_{\sim_Q}$ and the other in $[r']_{\sim_Q}$, i.e., the edges that \emph{cross} $[r]_{\sim_Q}$ and $[r']_{\sim_Q}$. In~\cite{firmani2016online} (see Equation (2) with our notation), $\mu_Q(r,r')$ is defined as the maximum edge weight in $\textit{Cross}_Q(r,r')$, that is
\begin{equation}\label{eq:crossing_DD}
     \mu_Q^{\max}(r,r') = \max_{e\in \textit{Cross}_Q(r,r')}w(e).
\end{equation}
We believe that the maximum is not representative enough of the aggregate matching evidence between $r$ and $r'$.
So we introduce a \emph{mean} version:
\begin{equation}\label{eq:crossing_density}
     \mu_Q^{\text{mean}}(r,r') = \frac{\sum_{e\in \textit{Cross}_Q(r,r')}w(e)}{|r|_Q\cdot|r'|_Q}.
\end{equation}

%After a sequence of queries $Q$, we denote by $[r]_Q$ (or, simply, $[r]$) the set of records that matches $r$ after the queries in $Q$, and by $|r|$ the cardinality of $[r]$. Given two representative records $r$ and $r'$, . We also define the \emph{crossing density between $r$ and $r'$} as

Let us briefly explain Equation~\ref{eq:crossing_density}. Note that $\mu_Q^{\text{mean}}(r,r')\in[0,1]$, since there are at most $|r|_Q\cdot|r'|_Q$ edges in $\textit{Cross}_Q(r,r')$. A value close to 1 indicates that $\textit{Cross}_Q(r,r')$ contains many edges with high weights. Conversely, a value close to 0 means that either there are few edges, their weights are low, or both. In Section~\ref{sec:experiments} we show that the presented algorithms perform better using $\mu_Q^{\text{mean}}$ than $\mu_Q^{\text{max}}$ in Equation~\ref{eq:benefit}. Note that using $\mu_Q^{\text{mean}}$ implies that $\benefit_Q(r,r') = \sum_{e\in \textit{Cross}_Q(r,r')}w(e)$.

\subsubsection{A greedy algorithm for the Heaviest Subgraph Problem}\label{subsub:greedy}

As we will explain in Subsection~\ref{subsub:current_community_batches}, at each step of our algorithm we need to select a subset of $\window$ records that maximizes the sum of benefit edges. This corresponds to solving an instance of Problem~\ref{prob:k-heaviest}, which, as previously stated, is NP-hard. Therefore, we use a greedy approach for Problem~\ref{prob:k-heaviest}.
If $b=2$, we simply select the two vertices of the edge with the maximum weight in $G$, which is the optimal solution. Otherwise, we start with $S=\{v\}$, where $v$ is the vertex whose incident edges have the maximum total weight. Formally, we define $w_e(x) =\sum_{e\in E(G), x\in e}w(e)$, thus $v = \argmax_{x\in V(G)} w_e(x)$.
Then, until $|S| = b$, we iteratively either add the vertex that maximizes the sum of the edge weights, or add the two vertices of the edge with maximum weight in the graph $G\setminus G[S]$. The pseudocode is reported in Algorithm~\ref{alg:greedy}.

\begin{algorithm}[h]
\caption{\textit{GreedyHS}$(G, b)$}\label{alg:greedy}
\KwIn{an edge-weighted graph $G$ and an integer $b$}
\KwOut{a suboptimal solution to Problem~\ref{prob:k-heaviest}}
{
\If{$b = 2$}{\label{line:b=2}
    let $e = (u,v)$ be the edge with the maximum weight\;
    \Return{$\{u,v\}$}
}
$u = \argmax_{x\in V(G)}\ w_e(x)$\;
$v = \argmax_{y \textit{ neighbor of } u}\ w((u,y))$\;
$S = \{u,v\}$\;
\While{$|S|< b$ and $|S|<|V(G)|$}{
    let $(x,y)$ be the heaviest edge in $G\setminus G[S]$\;
    let $z$ be the vertex in $G$ that maximizes $W = w(G[S \cup \{z\}])$\;
    \uIf{$w(G[S]) + w((x,y)) > W$ and $|S|\leq b - 2 $}{
        add $x$ and $y$ to $S$\;
    }
    \Else{
        add $z$ to $S$\;
    }
}
\Return{$S$}
}
\end{algorithm}

\subsubsection{Current-batches and community-batches}\label{subsub:current_community_batches}

We now have all the necessary definitions to describe our approach in detail. As mentioned at the beginning of this section, at each step we decide whether it is more convenient to try to enlarge already visited entities or to explore new records within the same community. This decision is guided by introducing two subsets of the records: \emph{current} and \emph{unqueried}.

At each step, only a subset of the records, called \emph{current}, is considered. The subset \emph{current} follows the sequence of heavy communities $c_1, c_2,\ldots$: at the first query \emph{current} is equal to $c_1$; once all records in $c_1$ have been queried at least once to the \oracle, \emph{current} is equal to $c_1\cup c_2$; once all records in $c_2$ have been queried, then \emph{current} is equal to $c_1\cup c_2\cup c_3$; and so on.

The set \emph{unqueried} corresponds, at each step of the algorithm, to the subset of records in \emph{current} that have not yet been queried to the \oracle. We now define two specific \emph{benefit} graphs.

\begin{definition}\label{def:g_benefit}
    Given a sequence of queries $Q$ and a set of records $V$, we define the \emph{benefit graph on $V$}, $B_{V}(Q)$, as the graph on $V/\sim_Q$ in which edge weights represent the benefit among records. Moreover, given a temperature $t$, we define $B_{V}^t(Q)$ as the subgraph of $B_{V}(Q)$ induced by all edges whose weight is higher than $t$.
\end{definition}

In Definition~\ref{def:g_benefit}, if $Q$ is clear from the context, we omit it. At each step, we aim to maximize the sum of benefits by considering either all records in \emph{current} or only the subset \emph{unqueried}. Considering \emph{current} corresponds to try to enlarge already visited entities, which is done by $\textit{GreedyHS}(B^t_{\textit{current}},\window)$, where the parameter $t$ is fixed in the next subsection. Considering \emph{unqueried} corresponds to visiting new records within the same community, which is done by $\textit{GreedyHS}(B_{\textit{unqueried}},\window)$.

For convenience, we refer to a batch composed of records in \emph{current} as a \emph{current-batch}, and a batch composed of records in \emph{unqueried} as a \emph{community-batch}. In the next subsection, we explain how to set the temperature $t$, which determines whether it is more advantageous to proceed with a \emph{current-batch} or a \emph{community-batch}.

A summary of the differences between \emph{current-batches} and \emph{community}\emph{-batches} is reported in Table~\ref{tab:batch_comparison}.
In particular, \emph{current-batches} are not disjoint, meaning that the same record may appear in multiple current-batches; records may belong to different communities; the number of discovered match edges is unbounded, as it depends solely on the (a priori unknown) ground truth of the dataset; and records may have already been queried.
In contrast, \emph{community-batches} are disjoint; all records belong to the same community; there exists an upper bound on the number of possible discovered match edges; and all records are unqueried.

\begin{table}[h]
    \small
    \setlength{\tabcolsep}{3pt}
    \centering
    \renewcommand{\arraystretch}{1.5}
    \begin{tabular}{|c|c|c|c|c|}
        \hline
        Type & Disjoint & Communities & \# Matches  & Queried/Unqueried\\
        \hline
        Current & No & Any & Unlimited & Both\\
        \hline
        Community & Yes & Same & $\leq \frac{\window(\window-1)}{2}$ & Unqueried\\
        \hline
    \end{tabular}
    \caption{Comparison between current-batches and community-batches.}
    \label{tab:batch_comparison}
\end{table}

\subsubsection{Setting the temperature}\label{subsub:temperature}

We now describe how to set the parameter $t$, which depends on the evolution of the benefit values. However, the benefit's value depends on the input dataset we are examining.
For instance, if the input dataset has many entities with a huge number of duplicated records, then it is possible to have high benefit values; conversely, if each entity of the input dataset has at most a few duplicates, then it is impossible to have high benefit values. In a nutshell, there is no single value of $t$ that works for all input datasets.
Therefore, we introduce two criteria to determine when the value of $t$ is too high and when it is too low.

Since we consider only two kinds of batches -- current-batches and community-batches -- the only alternative to querying the \oracle with a current-batch is to query it with a community-batch.
Thus, the criterion to determine when the value of $t$ is too low is to check whether the match edges discovered through a current batch are greater than the average number of match edges discovered through previous community batches.
To formalize this criterion, we introduce the following notation. Let $M^I_{comm}$ be the number of match edges discovered after querying the \oracle with the first $I$ community-batches, and let $\tau^I_{comm} = M^I_{comm}/I$ be the average of match edges discovered. If no confusion arises, then we omit the index $I$.
Let $M$ be the number of match edges discovered after querying the \oracle with a current batch. The criterion is: if $M < \tau_{comm}$, then $t$ is too low; therefore we double it (see Algorithm~\ref{alg:main} line~\ref{line:if_then_double}).

The criterion for determining when the value of $t$ is too high is to check whether $B^t_{\textit{current}}$ contains at least $\window$ vertices.
If not, we decrease the value of $t$ by multiplying it by the factor $(1-\frac{1}{\window})$ (see Algorithm~\ref{alg:main} line~\ref{line:else_5_percent}).
This ensures that we query the \oracle with a current-batch approximately every $\window$ queries.
Through experiments, we observed that our algorithm's performance improves when the frequency of current-batches is inversely proportional to $\window$, which motivates the factor $(1-\frac{1}{\window})$.

%This ensures that we avoid stopping querying the \oracle with current-batches in the case the values of benefit are low. Note that if the value of $t$ becomes too low, then it is doubled by the previous criterion.

By having the previous two criteria for high and low values of $t$, the choice of the initial value of $t$ is not relevant.
By default, we choose $t = \window$ at the beginning of the algorithm. Our main algorithm \perbacco is reported in Algorithm~\ref{alg:main}, where $\query(S)$ means ``query the \oracle with the records in $S$''.

\begin{algorithm}[h]
\caption{\perbacco}\label{alg:main}
\KwIn{a set of records $R$ and an integer $\window$}
\KwOut{a sequence of $\window$-batches of $R$}
{
Compute a similarity graph $G$ on $R$\;
Compute set $\mathcal{C}$ of heavy communities in $G$ by Algorithm~\ref{alg:community}\;
$t= b$\;
$\textit{current} = \emptyset$\;

\For{$c \in\mathcal{C}$}{
    $\textit{current} = \textit{current} \cup c$\;
    $\textit{unqueried} = c$\;
    \While{$|unqueried| \geq \window$}{
        $\textit{community\_batch}= \textit{GreedyHS}(B_{\textit{unqueried}},\window)$\;
        $\query(\textit{community\_batch})$\;
        $\textit{unqueried} = \textit{unqueried} \setminus \textit{community\_batch}$\;
        \While{$|V(B^t_{\textit{current}})| \geq\window$}{
        $\textit{current\_batch} = \textit{GreedyHS}(B_{\textit{current}}^t,\window)$\;
        $\query(\textit{current\_batch})$, and let $M$ be the number of match edges discovered\;
        $unqueried = unqueried \setminus \textit{current\_batch}$\;
        \If{$M < \tau_{comm}$}{\label{line:if_then_double}
            $t= 2t$\;
                }
        }
        $t = \big(1-\frac{1}{\window}\big)\cdot t$\label{line:else_5_percent}\;
    }
}
$\textit{current} = V(G)$\;
\While{there are non-inferable edges}{
    $\textit{current\_batch} = \textit{GreedyHS}(B_{\textit{current}}^t,\window)$\;
    $\query(\textit{current\_batch})$\;
}
}
\end{algorithm}

\section{Experimental Evaluation}\label{sec:experiments}

In this section, we discuss the results of our experimental evaluation. We compare all algorithms on publicly available datasets and synthetic ones. The datasets used in our experiments are described in Subsection~\ref{sub:datasets}.
To the best of our knowledge, there are no known algorithms for progressive entity resolution via adaptive batch query; so we extend the algorithms in~\cite{firmani2016online} from the pairwise query case to the batch query case, see Subsection~\ref{sub:competing_algorithms}.
We omit several non-adaptive competitor algorithms, such as the ones in~\cite{DBLP:journals/csur/ChristophidesEP21,crowder,papadakis_exploration},  since their performance is strictly lower than that of competing algorithms described in Subsection~\ref{sub:competing_algorithms}. Indeed, adapting the queries to the \oracle based on previous answers provides a significant advantage in terms of progressive recall.
All algorithms are implemented in Python, in a common framework\footnote{\url{https://github.com/Stravanni/pERbacco}}, where ties are broken randomly. We ran experiments on a machine with an AMD Ryzen 9 5950X CPU, 64 GB RAM, and an NVIDIA RTX 3090 GPU under Ubuntu 22.04.

In our experimental evaluation, the similarity graph is built using \emph{standard-blocking} and \emph{meta-blocking} in the JedAI library~\cite{DBLP:journals/is/Mandilaras0GSTG21}.
For each dataset, we tune the meta-blocking parameters to achieve a recall above 0.9, ensuring that the majority of true match edges were preserved in the resulting graph.
This choice reflects a concrete instantiation of the similarity graph abstraction introduced in Subsection \ref{sub:similarity_graph} and does not affect the generality of the proposed approach.

We report progressive recall, i.e., the fraction of ground-truth match edges discovered after each query budget. Under our consistent-oracle assumption, every discovered positive match is correct, so the precision of discovered matches is 1 and the corresponding F1-score is $2\cdot recall/(1+recall)$, a monotone transformation of recall. If the oracle is noisy, precision and F1 become essential metrics; modeling that setting is orthogonal to the batch-scheduling problem studied here.

Regarding the running time, the blocking and CDAs require from a few seconds on small datasets to about ten minutes on the largest ones. The algorithm \perbacco and its competitors (see Section~\ref{sub:competing_algorithms}) require at most a few seconds per query when using the ground truth as the \oracle. When using an LLM (in our case, GPT-5 mini), the time required to compute Figure~\ref{fig:run-llm} on \cora dataset ranges from 2 seconds for $b=2$ to 20 seconds for $b=20$.

Consequently, since blocking and CDAs are executed only once, the main bottleneck is determined by querying the \oracle for each query. In addition to execution time, the monetary cost of using an LLM must also be considered. For this reason, it is important to study entity resolution in a progressive manner, showing that, at the same level of recall, our approach requires fewer calls to the \oracle.

\begin{comment}
{\color{red} commento sull'ordine di grandezza dei tempi.

- blocking from seconds to minutes

- cda idem sopra

- algo (ancora meno)

- quando usaiamo LLM (gpt-3 mini) "vedi Figura con due phi", ongni cihamata va da ~2sec (b=2) ad ~15sec (b=20), quindi per arrivare anche solo a 1phi ci vuole (conti a spanne), quindi 1 (2?) odg in piu'

- quindi frase donatella
}
\end{comment}

\subsection{Competing Algorithms}\label{sub:competing_algorithms}

Firmani et al.~\cite{firmani2016online} presented two algorithms for the progressive entity resolution with pairwise adaptive query. The two algorithms are called $s_{\text{edge}}$ and $s_{\text{hybrid}}$; we do not deeply explain these algorithms, and we refer to~\cite{firmani2016online} for further details.
In a nutshell, $s_{\text{edge}}$ queries the \oracle each time with the edge having the highest benefit, and $s_{\text{hybrid}}$ queries the \oracle with the edges adjacent to the vertex with the highest total benefit (i.e., the vertex whose sum of benefits of adjacent edge is maximized), until a match edge is discovered or a prefixed number of edges have been queried.
We merge both algorithms in Algorithm~\ref{alg:DD}: when $\window = 2$, we obtain exactly $s_{\text{edge}}$ because of the case $b=2$ in Line~\ref{line:b=2} of Algorithm~\ref{alg:greedy}; when $\window > 2$, we adapt $s_{\text{hybrid}}$ to the batch-queries case by using the $\textit{GreedyHS}$ algorithm.

In our experiments, we denote by \perbac (resp., \DD) the results of Algorithm~\ref{alg:DD} when the benefits are computed using $\mu^{\text{mean}}$ (resp., $\mu^{\max}$). Note that \perbac coincides with \perbacco when there are no heavy communities. Moreover, in~\cite{firmani2016online} (titled \emph{Online Entity Resolution Using an Oracle}), the benefits are computed with $\mu^{\max}$. These observations motivate the naming of the algorithms.

\begin{algorithm}[h]
\caption{batched ER without communities}\label{alg:DD}
\KwIn{a set of records $R$ and an integer $\window$}
\KwOut{a sequence of $\window$-batches of $R$}
{
Compute a similarity graph $G$ on $R$\;
%Consider all the records as visited\;
\While{there are non-inferable edges}{
    $\query(\textit{\textit{GreedyHS}}(B_{V(G)},\window))$\;}
}
\end{algorithm}

\subsection{Datasets and Query Complexity}\label{sub:datasets}

We evaluate \perbacco and competing algorithms on five real datasets, spanning domains such as product, business, and citation, and on synthetic datasets. The \cora dataset consists of 1,295 bibliographic records of machine-learning publications~\cite{firmani2016online}. \camera contains 29.8$k$ specifications of real-world cameras collected from 24 e-commerce websites, originally used in the ACM SIGMOD 2020 programming contest~\cite{crescenzi2021alaska}. \funding\footnote{https://raw.githubusercontent.com/qcri/data\_civilizer\_system/master/grecord\_
service/gr/data/address/address.csv} includes 16.3$k$ records on financing requests submitted to the NYC Council Discretionary Funding. The Web Data Commons Products dataset~\cite{peeters2023wdc}\footnote{https://webdatacommons.org/largescaleproductcorpus/wdc-products/index.html\#toc5} contains product offers extracted from 3,259 web shops in 2020 via schema.org annotations; we use the \wdc variant, comprising 3,841 offers that describe 1,000 unique real-world products. Finally, the \voters\footnote{https://hpi.de/naumann/projects/repeatability/datasets/ncvoters-dataset.html} dataset contains demographic information on 14.2$k$ registered voters from North Carolina, grouped by sex and race~\cite{voters_data}.
We selected these datasets because they exhibit different entity size distributions (see Table~\ref{tab:dataset_table} that summarizes key statistics for each dataset). Moreover, we did not consider one-to-one datasets, as they are not suitable for studying entity resolution in a progressive manner via batch queries.

The synthetic datasets are randomly generated according to a specific distribution of weights and edges. The distribution of entity sizes follows the entity size distribution observed in the \camera dataset, which exhibits power-law behavior; to avoid excessively large entities, we impose a maximum size of 300.
We generate a dataset, \synth, consisting of approximately 39k records and exactly 10k entities. The recall is fixed at 0.95, while we create three variants with precision values of 0.5, 0.2, and 0.05, respectively. The missing match edges and the non-match edges are generated randomly.

To model edge weights, we analyzed the similarity graphs of the real datasets (Subsection~\ref{sub:similarity_graph}): non-match weights consistently follow a power-law distribution, whereas match weight distributions vary substantially across datasets, so we generate the latter using a uniform distribution.

\setlength{\tabcolsep}{3pt}
{\small
\begin{table}[!h]
    \centering
    \begin{tabular}{|c|c|c|c|c|c|c|}
        \hline
        \multirow{2}{*}{Dataset} & \multirow{2}{*}{Records} & \multirow{2}{*}{Matches} & \multirow{2}{*}{Entities} & \multicolumn{3}{c|}{Entity Sizes} \\
        \cline{5-7}
        & & & & Mean & Med. & Max \\
        \hline
        \cora & 1295 & 17184 & 93 (112) & 13.7 & 7  & 64\\
        \camera & 29787  & 541411 & 2087 (9005) & 9.9  & 4  & 256   \\
        \funding & 16258  & 135122 & 2351 (3113) & 6.5  & 4  & 115   \\
        \wdc & 3841  & 8971 & 1000 (1000) & 3.5  & 2  & 11  \\
        \voters & 14183  & 9819 & 5657 (6692) & 2.3  & 2   & 8   \\
        %cddb & 9763  & 300  & 221 (9497) & 2.1  & 2   & 6  \\
        %census & 841  & 344 & 333 (504) &  2.0  & 2   & 4   \\
        %restaurant & 864  & 112 & 112 (749) & 2   & 2  & 2   \\
        \hdashline
        \synth & 39161  & 915128  & 3890 (10000) & 8.5 & 3  & 257   \\
        %synth\_5000 & 19499  & 458786  & 1905 (5000) & 8.6 & 3  & 256   \\
        %synth\_1000 & 4793  & 171835  & 386 (1000) & 8.5 & 3  & 256   \\
        %synth\_250 & 1171  & 33901  & 105 (250) & 9.7 & 3  & 178   \\

        \hline
    \end{tabular}
    \caption{Dataset statistics. The Entities column reports the number of entities with size at least 2, with the total number of entities shown in parentheses. The Entity Sizes column excludes entities of size 1.}
    \label{tab:dataset_table}
\end{table}
}

Table~\ref{tab:upper_lower_bound} reports the lower and upper bounds of $\Phi_{10}$, as defined in Equation~\ref{eq:phi_b}, denoted by $\phi_{10}$ and $\overline{\Phi}_{10}$, respectively.
The upper bound $\overline{\Phi}_{10}$ is computed using the Python library \texttt{binpacking}\footnote{\url{https://pypi.org/project/binpacking/}}, which successfully finds the exact solution for all instances of Problem~\ref{prob:bin} considered in our experiments.
We observe that the lower and upper bounds differ by at most 3\%, indicating that the lower bound provides a tight approximation of $\Phi_{10}$.

{\small
\begin{table}[h!]
\centering
\begin{tabular}{c c c c c c c }
\hline
Bound & \cora & \camera & \funding & \wdc &\voters & \synth\\
\hline
$\phi_{10}$          & 137  & 2436  & 1612 & 391 & 1315 & 3514 \\
$\overline\Phi_{10}$ & 137  & 2455  & 1640 & 396 & 1354 & 3544 \\
\hline
\hline
error & 0\% & 0.7\% & 1.7\% & 1.2\% & 2.9\% & 0.8\%\\
\hline
\end{tabular}
\caption{Upper and lower bound of $\Phi_{10}$ for all datasets.}
\label{tab:upper_lower_bound}
\end{table}
}

\subsection{A Priori Selection of CDAs}\label{sub:communities_experiments}

We tested four edge-weighted CDAs: Louvain~\cite{louvain}, Leiden~\cite{leiden}, Asynchronous Label Propagation~\cite{lpa}, and Infomap~\cite{infomap}; see~\cite{survey_cda} for a survey. The only parameter we expose in the main text is the heaviness threshold $\lambda_w$. Empirically, Louvain and Leiden most often yield the highest record ratio, and low $\lambda_w$ values help dense datasets because they preserve more candidate matches inside heavy communities. On sparse datasets, the same choice can add many non-match edges, so \perbac, which does not use communities, may be preferable.

Table~\ref{tab:general_trends} illustrates this behavior on one dense dataset (\funding) and one sparse dataset (\voters), using $\window=10$. The table is not used to tune per-dataset parameters in the main comparison: unless stated otherwise, we use Louvain with $\lambda_w=0.05$.

{\small
\begin{table}[h]
\centering
\begin{tabular}{l c c c c}
\toprule
\textbf{Dataset} & $\lambda_w$ & \textbf{CDA} & \textbf{record ratio} & \textbf{recall} \\

\midrule
\multirow{6}{*}{\funding}
  & \multirow{2}{*}{0.05}  & Louvain  & 0.98 & \textbf{0.884}\\
  &                        & Leiden  & 0.71  & 0.856   \\
\cmidrule(lr){2-5}
  & \multirow{2}{*}{0.15}  & Louvain  & 0.83 & 0.843 \\
  &                        & Leiden  & 0.28  & 0.797 \\

  \cmidrule(lr){2-5}
 &   \multirow{1}{*}{//}   & w/o CDA  & //   & 0.817 \\
\midrule

\multirow{6}{*}{\voters}
  & \multirow{2}{*}{0.05}  & Louvain & 0.80  & 0.232  \\
  &                        & Leiden  & 0.10  & \underline{0.280} \\
\cmidrule(lr){2-5}
  & \multirow{2}{*}{0.15}  & Louvain & 0.08  & \textbf{0.289}  \\
  &                        & Leiden  & 0.01  & \underline{0.284}\\
  \cmidrule(lr){2-5}
 &   \multirow{1}{*}{//}   & w/o CDA  & //   & \underline{0.284} \\
\bottomrule
\end{tabular}
\caption{Performance of \perbacco for $\window=10$, on different datasets, CDAs and values of $\lambda_w$. When no CDA is used, then \perbacco reduces to \perbac. Bold indicates the best combination of CDA and $\lambda$ for each dataset, and underlined results are close to the best. The recall is computed after $\phi_{10}$ queries.}
\label{tab:general_trends}
\end{table}
}

\begin{comment}
\subsection{Complexity analysis}
\begin{itemize}
    \item il costo di una singola domanda consiste in: calcolare i benefit (e questo costa $O(\window^2 * n)$, basta vedere il Lemma 4 di DD), ordinare i benefit e eliminarne alcuni (in teoria posso avere $O(\window^2 * n)$ aggiornamenti al database, che può essere di dimensione $O(n^2)$ nel caso peggiore, quindi tenerlo ordinato può costare parecchio), a lanciare il GreedyHS.
    \item in generale noi abbiamo molti meno spigoli di Donatella-Divesh per via dell'insieme current che cresce poco alla volta. Infatti il numero degli spigoli "esplode" quando arrivo all'ultimo while.
    \item inoltre l'analisi worst-case ha ancora meno senso perché è improbabile che tutti i vertici hanno  grado $n$. Però nel caso peggiore un vertice può avere caso $n$ anche se parto da grado limitato
    \item forse direi semplicemente che DD non aggiornano sempre tutti i benefit, ma solo le probabilità, e poi calcolano i benefit dei primi $n$ spigoli con probabilità maggiore. Questo noi NON possiamo farlo perché usando la media si creano nuove probabilità (con il massimo questo non accade), e quindi queste nuove probabilità devono comunque essere ordinate. Dire anche che negli esperimenti abbiamo implementato DD ordinando i benefit; questo dovrebbe migliorare i risultati rispetto a considerare i primi $n$ come detto nel loro paper.
\end{itemize}
\end{comment}

\subsection{Suboptimal Solution}\label{sub:suboptimal_solution}

As stated in Theorem~\ref{th:b_optimal_solution_not_exist}, a $b$-optimal solution does not exist in the general case, and selecting the batch that maximizes the estimated gain is NP-hard (Theorem~\ref{th:progressive_is_NP_hard}). Nevertheless, Theorem~\ref{th:optimal} characterizes a $b$-optimal solution for some special cases, and Proposition~\ref{prop:max_gain_with_HSP} shows that this solution maximizes the number of match edges at each query. Its implementation requires solving an instance of the NP-hard Heaviest Subgraph Problem for each query; we therefore call \emph{suboptimal} solution the algorithm of Theorem~\ref{th:optimal} where HSP is replaced by $\textit{GreedyHS}_{1000}$, where $\textit{GreedyHS}_{1000}(G,b)$ denotes $\textit{GreedyHS}(G_{1000},b)$, and $G_{1000}$ is the subgraph of $G$ induced by its 1000 heaviest edges.

If $\window = 2$, a 2-optimal solution exists, as discussed in Subsection~\ref{sub:optimal_solution}. For other values of $\window$, we do not have any theoretical guarantee about the performance of the suboptimal solution. Nevertheless, in all our experiments the suboptimal solution achieves a recall of at least 0.98  after $\phi_{\window}$ queries on every tested dataset.

\begin{figure*}[t!]
    \centering
    % --- Prima riga ---
    \begin{subfigure}{0.24\textwidth}
        \centering
        \includegraphics[width=\linewidth]{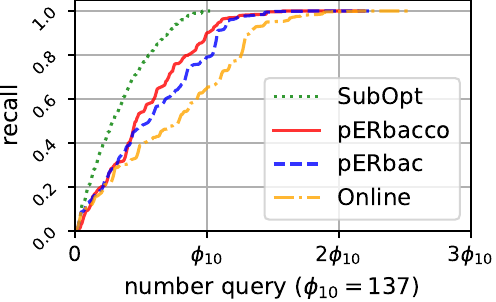}
        \caption{\cora}
    \end{subfigure}%
    \begin{subfigure}{0.24\textwidth}
        \centering
        \includegraphics[width=\linewidth]{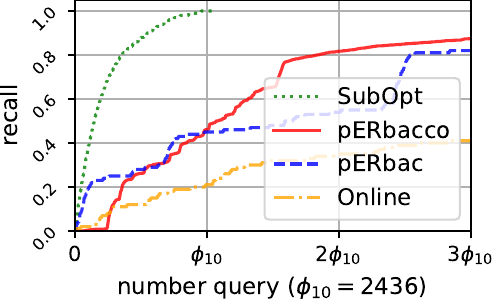}
        \caption{\camera}
    \end{subfigure}%
    \begin{subfigure}{0.24\textwidth}
        \centering
        \includegraphics[width=\linewidth]{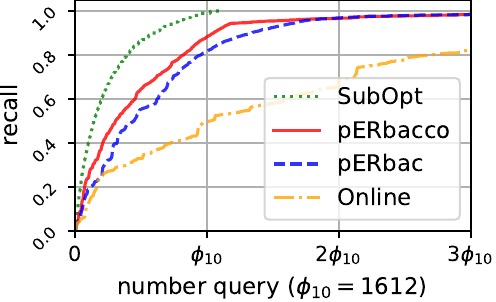}
        \caption{\funding}
    \end{subfigure}%
    \begin{subfigure}{0.24\textwidth}
        \centering
        \includegraphics[width=\linewidth]{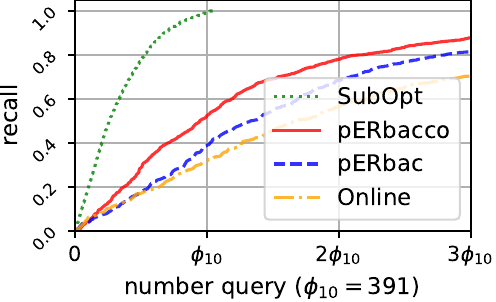}
        \caption{\wdc}
    \end{subfigure}

    \vspace{0.3cm} % spazio tra le righe

    % --- Seconda riga ---
    \begin{subfigure}{0.24\textwidth}
        \centering
        \includegraphics[width=\linewidth]{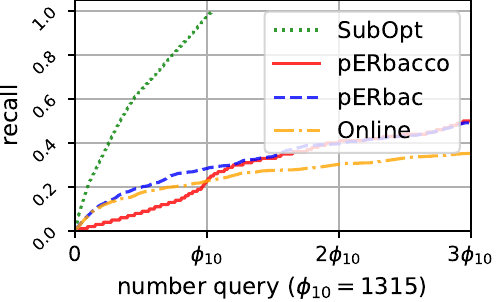}
        \caption{\voters}
    \end{subfigure}%
    \begin{subfigure}{0.24\textwidth}
        \centering
        \includegraphics[width=\linewidth]{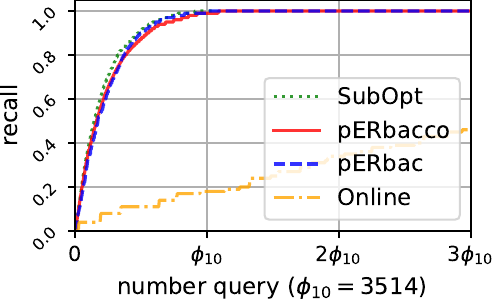}
        \caption{\synth, precision 0.5}
    \end{subfigure}%
    \begin{subfigure}{0.24\textwidth}
        \centering
        \includegraphics[width=\linewidth]{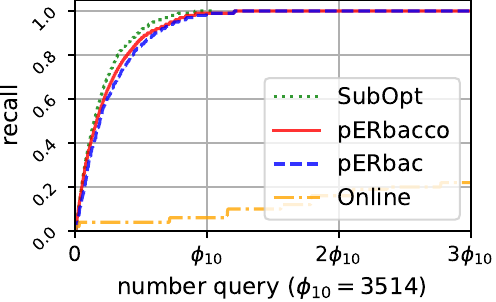}
        \caption{\synth, precision 0.2}
    \end{subfigure}%
    \begin{subfigure}{0.24\textwidth}
        \centering
        \includegraphics[width=\linewidth]{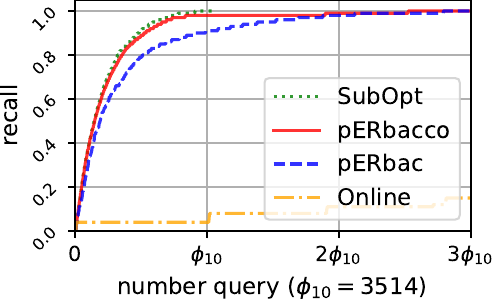}
        \caption{\synth, precision 0.05}
    \end{subfigure}

    \caption{Progressive recall on real and synthetic datasets with $\window = 10$.}
    \label{fig:progressive_recall_window=10}
\end{figure*}

\subsection{Evaluation}

We compare our approach \perbacco with the competing algorithms  \perbac and \DD described in Subsection~\ref{sub:competing_algorithms}, and with the suboptimal solution discussed in Subsection~\ref{sub:suboptimal_solution}, denoted by \subopt.

In Figure~\ref{fig:progressive_recall_window=10}, we compare the algorithms on datasets listed in Table~\ref{tab:dataset_table} with $\window$ equal to 10. For algorithm \perbacco, as discussed in Subsection~\ref{sub:communities_experiments}, we use the Louvain method as CDA, and $\lambda_w=0.05$. All algorithms are implemented with $\textit{GreedyHS}_{1000}$ in place of $\textit{GreedyHS}$ to further speed up the computation.%, as described in Subsection~\ref{sub:runtime_analysis}.

First, note that \perbac outperforms \DD on all datasets at each query, even though \perbac is obtained from \DD with a small improvement (i.e., $\mu^{\text{mean}}$ instead of $\mu^{\max}$ in Equation~\ref{eq:benefit}).

On the \funding and \wdc datasets, \perbacco outperforms both competitors at every query. On \cora, the algorithms perform similarly during the initial queries, after which \perbacco achieves higher recall. Note that, on \cora, all algorithms stop before $3\phi_{10}$ queries, as they resolve all edges in the similarity graph.

On \camera, \perbacco exhibits a slow start. Upon inspection, we observed that the initial heavy communities contain few match edges, despite having many heavy edges, which explains the slow initial behavior. As noted in Subsections~\ref{sub:graph_partitioning} and~\ref{sub:communities_experiments}, the construction of heavy communities is based on a heuristic approach. Nevertheless, after $\phi_{10}$ queries, the recall values achieved by \perbacco are higher than those of the competing algorithms.

\begin{figure}[t]
      \centering
      \includegraphics[width=0.95\columnwidth]{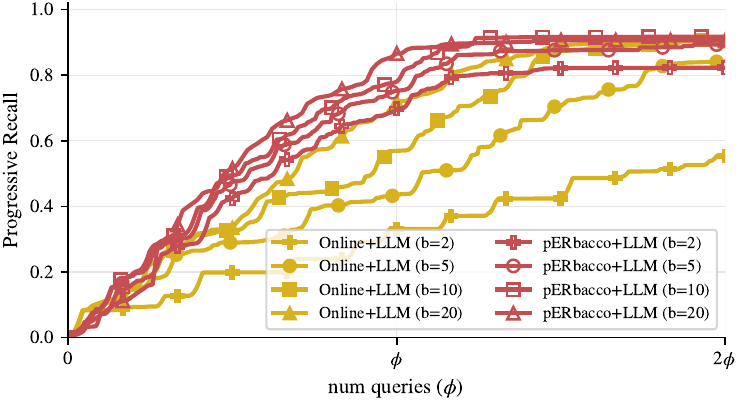}
      \caption{Progressive recall on \cora for \perbacco and \DD under the same LLM setting
  ({GPT-5 mini}, few-shot, 10 positive and 10 negative example pairs), for batch sizes $b \in
  \{2,5,10,20\}$. The x-axis reports the query budget normalized by $\phi_b$, so curves with different batch sizes are
  directly comparable up to $2\phi_b$.% (i.e., twice the number of invocations required to resolve the entire dataset).
  }
      \label{fig:run-llm}
  \end{figure}
% \begin{figure*}[t!]
%     \centering
%     % --- Prima riga ---
%     \begin{subfigure}{0.24\textwidth}
%         \centering
%         \includegraphics[width=\linewidth]{figures/funding_2.pdf}
%         \caption{\funding with $\window = 2$}
%     \end{subfigure}%
%     \begin{subfigure}{0.24\textwidth}
%         \centering
%         \includegraphics[width=\linewidth]{figures/funding_5.pdf}
%         \caption{\funding with $\window = 5$}
%     \end{subfigure}%
%     \begin{subfigure}{0.24\textwidth}
%         \centering
%         \includegraphics[width=\linewidth]{figures/funding_20.pdf}
%         \caption{\funding with $\window = 20$}
%     \end{subfigure}%
%     \begin{subfigure}{0.24\textwidth}
%         \centering
%         \includegraphics[width=\linewidth]{figures/funding_40.pdf}
%         \caption{\funding with $\window = 40$}
%     \end{subfigure}

%     \vspace{0.3cm} % spazio tra le righe

%     \caption{Progressive recall on \funding with $\window = 2,5,20$, and $40$.}
%     \label{fig:progressive_recall_funding}
% \end{figure*}

On \voters, algorithm \perbacco performs worse than \perbac because the dataset is sparse. We point out that this behavior could be predicted a priori based on the analysis reported in Subsection~\ref{sub:communities_experiments}.

On the synthetic dataset \synth, we observe that both \perbac and \perbacco are close to the suboptimal solution when the precision is 0.5 or 0.2. When the precision decreases to 0.05 -- which is more realistic at a recall of 0.95 -- \perbacco significantly outperforms \perbac.
It is worth noting that there is a huge difference between all algorithms in \camera and \synth with precision 0.05, although the distributions of communities, edges, and weights are similar. This is due to the randomness of non-match edges in \synth. In contrast, in the real dataset \camera, the similarity graph is obtained via distances between records, thus similar records that do not represent the same entity are likely to share a heavy edge.
This makes it more difficult to discover heavy communities composed mainly by match edges.

% Figure~\ref{fig:progressive_recall_funding} reports a comparison of the considered algorithms on the \funding dataset for $b = 2, 5, 20,$ and $40$. Analogous results on the remaining datasets are omitted due to space constraints.
% When $b = 2$, we consider the classical pairwise comparison setting. In this case, a 2-optimal solution exists, as established in Theorem~\ref{th:optimal} and in~\cite{firmani2016online}, and we denote it by \texttt{Opt}.
% For this setting, we note that the implementation of \DD coincides with the one presented in~\cite{firmani2016online}. Therefore, the results in the leftmost subfigure indicate that our approach improves the state of the art for pairwise Entity Resolution with an Oracle.

% Figure~\ref{fig:progressive_recall_funding} also shows that the algorithms exhibit similar behavior across different values of $b$, as the query axis is rescaled according to $\phi_b$. Overall, \perbacco consistently outperforms \DD at every query and almost always outperforms \perbac, highlighting the effectiveness of heavy communities.

Figure~\ref{fig:run-llm} reports an illustrative LLM run on \cora with GPT-5 mini, few-shot prompting, and $b \in \{2,5,10,20\}$, using the normalized query budget up to $2\phi_b$. \perbacco is largely insensitive to $b$, while \DD improves as $b$ increases; the plateau after $\phi_b$ is due to LLM errors. Handling such errors and optimizing prompt-level cost-quality tradeoffs~\cite{10597751, AvengER, nananukul2024cost} are complementary to our focus on budget-aware batch selection.

\begin{comment}

\begin{figure*}[!htb]
    \centering
    % --- Prima riga ---
    \begin{subfigure}{0.24\textwidth}
        \centering
        \includegraphics[width=\linewidth]{figures/cora_2.pdf}
        \caption{\cora}
    \end{subfigure}%
    \begin{subfigure}{0.24\textwidth}
        \centering
        \includegraphics[width=\linewidth]{figures/camera_2.pdf}
        \caption{\camera}
    \end{subfigure}%
    \begin{subfigure}{0.24\textwidth}
        \centering
        \includegraphics[width=\linewidth]{figures/funding_2.pdf}
        \caption{\funding}
    \end{subfigure}%
    \begin{subfigure}{0.24\textwidth}
        \centering
        \includegraphics[width=\linewidth]{figures/wdc80_2.pdf}
        \caption{\wdc}
    \end{subfigure}

    \vspace{0.3cm} % spazio tra le righe

    % --- Seconda riga ---
    \begin{subfigure}{0.24\textwidth}
        \centering
        \includegraphics[width=\linewidth]{figures/voters_2.pdf}
        \caption{\voters}
    \end{subfigure}%
    \begin{subfigure}{0.24\textwidth}
        \centering
        \includegraphics[width=\linewidth]{figures/synth_250_2.pdf}
        \caption{synth250}
    \end{subfigure}%
    \begin{subfigure}{0.24\textwidth}
        \centering
        \includegraphics[width=\linewidth]{figures/synth_1000_2.pdf}
        \caption{synth1000}
    \end{subfigure}%
    \begin{subfigure}{0.24\textwidth}
        \centering
        \includegraphics[width=\linewidth]{figures/synth_5000_2.pdf}
        \caption{synth5000}
    \end{subfigure}

    \caption{\blue{\window = 2}}
    \label{fig:8plots}
\end{figure*}
\end{comment}

\section{Related work}\label{sec:related_work}

\noindent\textbf{Pay-as-you-go data integration.}
The pay-as-you-go (a.k.a.\ progressive) approach to data integration was introduced to address scenarios in which integrating all data upfront is infeasible or impractical~\cite{DBLP:conf/cidr/MadhavanCDHJKY07}.
In this line of research, progressive methods for entity resolution (ER) prioritize candidate record pairs according to their estimated likelihood of being true matches, to discover as many matches as early as possible~\cite{DBLP:journals/tkde/WhangMG13, papadakis_exploration, DBLP:journals/tkde/SimoniniPPB19, DBLP:journals/tkde/PapenbrockHN15, DBLP:journals/pvldb/SimoniniZBN22}.
Active-learning ER is related because it scores candidate pairs, but its goal is to minimize labeled pairs for training a supervised matcher rather than to schedule oracle calls for progressive resolution quality~\cite{papadakis_exploration, DBLP:conf/kdd/SarawagiB02, DBLP:journals/pacmmod/GenossarGS23}. To the best of our knowledge, we propose the first study of a pay-as-you-go ER setting in which progress is achieved by querying an oracle over \emph{bounded-size sets} of records, rather than individual pairs.

\noindent\textbf{ER using an oracle.}
In the crowdsourcing literature, several works have explored ER abstractions that go beyond purely pairwise queries~\cite{waldo, crowder}.
These approaches are motivated by human-in-the-loop settings, where an oracle (the crowd) is queried to determine match relationships among small groups of records.
However, their setting is substantially different from ours. The hybrid approach in~\cite{waldo} relies on both pairwise and batch queries. Pairwise comparisons are considered more reliable and are used to resolve pairs with discordant answers across different batch queries.
The method in~\cite{crowder} processes all the batches in parallel. As a result, answers obtained from earlier queries cannot be exploited to guide subsequent batch selection, preventing a truly progressive, pay-as-you-go resolution process. Moreover, this approach does not study how to adaptively select batches under a budget, nor does it provide a formal, model-independent oracle abstraction.
%Existing crowdsourcing-based methods typically define batches upfront and process them independently or in parallel.
%As a result, answers obtained from earlier queries cannot be exploited to guide subsequent batch selection, preventing a truly progressive, pay-as-you-go resolution process.
%Moreover, these approaches do not study how to adaptively select batches under a budget, nor do they provide a formal, model-independent oracle abstraction.
Similarly, we do not consider noisy or inconsistent oracle answers in this work.
Handling errors arising from imperfect models or human annotation is a challenging problem in its own right and has been studied extensively in the context of crowdsourced ER, where the oracle abstraction is comparable (e.g.,~\cite{crowd_error_01,crowd_error_02,crowd_error_03,crowd_error_04}).
In this work, we therefore abstract away from oracle errors to focus on the algorithmic problem of budget-aware batch selection.

\noindent\textbf{ER with LLMs and Set-based models.}
Besides the seminal paper~\cite{DBLP:journals/pvldb/NarayanCOR22}, Peeters et al.~\cite{DBLP:conf/edbt/PeetersSB25} provide a comprehensive analysis of LLM-based ER, which is primarily based on pairwise entity matching.
More recently, \textsc{LLM-CER}~\cite{LLM-CER} and \textsc{ComME}~\cite{DBLP:conf/coling/WangCLCHSWZ25} show that LLMs can effectively act as \emph{bounded oracles} for ER by jointly clustering small sets of records within a single, carefully designed prompt, thus making multi-record ER feasible.
However, these works focus on the behavior of individual queries and do not address how to allocate a limited budget of LLM calls at the dataset level.

The idea of resolving entities by jointly analyzing sets of records is not specific to LLMs.
Set-based models such as Set Transformers~\cite{DBLP:conf/icml/LeeLKKCT19}, as well as vision-based identity clustering systems~\cite{DBLP:conf/cvpr/SchroffKP15}, directly operate on bounded collections of items and infer clusters without reducing the task to independent pairwise comparisons.

Our proposal is complementary to these lines of work.
Rather than introducing a new oracle implementation, we study ER as a dataset-level inference problem under budget constraints, focusing on how to progressively resolve an entire dataset by adaptively selecting bounded-size oracle queries.

\section{Conclusions}\label{sec:conclusions}

% In this paper, we studied progressive batched entity resolution, a formulation allowing an oracle to resolve batches of records instead of pairs.
% We showed that, unlike the pairwise case, an optimal sequence of batch queries does not always exist, and that selecting batches to maximize the number of newly discovered matches is NP-hard. These results formally establish the increased computational complexity of allocating oracle queries at the dataset level and generalize the results for pairwise ER.

% Our experimental evaluation on both real-world and synthetic datasets demonstrates that the proposed approach consistently achieves higher progressive recall than state-of-the-art baselines. Future work includes extending the framework to heterogeneous or noisy oracles.

This paper studies entity resolution in a pay-as-you-go setting where an oracle jointly resolves bounded-size batches of records under a limited budget. This abstraction captures a wide range of modern resolution mechanisms, including set-based~\cite{DBLP:conf/icml/LeeLKKCT19, DBLP:conf/cvpr/SchroffKP15} models, human-in-the-loop systems~\cite{crowder, waldo}, and LLMs~\cite{DBLP:conf/coling/WangCLCHSWZ25}, while remaining agnostic to the specific underlying model. We show that moving beyond pairwise queries fundamentally changes the problem: an optimal sequence of batch queries may not exist, and selecting the next batch by maximizing the estimated gain is NP-hard. These results clarify why strategies designed for pairwise ER~\cite{firmani2016online} do not directly extend to the batched setting.

Building on this analysis, we proposed a practical approach for \emph{progressive batched entity resolution} that adaptively allocates oracle calls by exploiting similarity structure and previously acquired information.
Experiments on real and synthetic datasets demonstrate that this strategy consistently achieves higher progressive recall than state-of-the-art baselines under comparable budgets, highlighting the benefits of explicitly reasoning over batches.

As a final remark, we note that our solution assumes a consistent oracle to isolate the algorithmic problem of budget-aware batch selection.
Extending the framework to handle noisy or erroneous oracle outputs, incorporating uncertainty into batch selection, and accounting for heterogeneous oracle costs are important, non-trivial directions for future research.

\section*{Acknowledgment}
This work is partly funded by the HORIZON Research and Innovation Action 101135576 INTEND ``Intent-based data operation in the computing continuum''.
The authors used OpenAI ChatGPT/Codex to assist with grammar polishing, phrasing suggestions, LaTeX and code-generation support, and automation of experiment-running scripts. The AI-assisted editing was applied throughout the manuscript to text written by the authors; the code and automation assistance supported the experiments reported in Section~\ref{sec:experiments}. The authors reviewed, verified, and approved all content, code, experimental results, claims, and conclusions, and no AI-generated text was used as substantive scientific content.

%\clearpage

\bibliographystyle{IEEEtran}
\bibliography{bib}

\end{document}